\documentclass[conference]{IEEEtran}
\usepackage{psfig,amsmath,amssymb ,algorithm,amsthm,,cite,cases,bm,float,subfigure,graphicx}

\newtheorem{defn}{Definition}[section]
\newtheorem{thm}{Theorem}[section]
\newtheorem{lem}{Lemma}[section]

\newcommand{\vs}{\vspace{-0pt}}
\newcommand{\beq}{\begin{equation}}
\newcommand{\eeq}{\end{equation}}  
\newcommand{\eea}{\end{align}}
\newcommand{\bea}{\begin{align}}  
\newcommand{\x}{{\mathbf{x}}  }
\newcommand{\li}{\left<}
\newcommand{\ri}{\right>}

\newcommand{\D}{{\mathbf{D}}}

\newcommand{\ab}{\mathbf{a}}

\newcommand{\Lc}{\mathcal{L}}
\newcommand{\hX}{\mathbf{\hat{\X}}}
\newcommand{\hx}{\mathbf{\hat{\x}}}
\newcommand{\bb}{{\mathbf{b}}}
\newcommand{\cb}{{\mathbf{c}}}
\newcommand{\db}{{\mathbf{d}}}
\newcommand{\rb}{\mathbf{r}}
\newcommand{\Ib}{{\mathbf{I}}}
\newcommand{\ib}{{\mathbf{i}}}
\newcommand{\Scb}{{\bar{{\mathcal{S}}}}}
\newcommand{\SB}{{\mathbf{S}}}
\newcommand{\Wb}{\mathbf{W}}
\newcommand{\Lcb}{\bar{\mathcal{L}}}

\newcommand{\vv}{{\mathbf{v}}}

\newcommand{\X}{\mathbf{X}}

\newcommand{\Pro}{{\mathbb{P}}}
\newcommand{\y}{{\mathbf{y}}}
\newcommand{\s}{{\mathbf{s}}}
\newcommand{\z}{{\mathbf{z}}}
\newcommand{\A}{{\mathbf{A}}}
\newcommand{\B}{\mathbf{B}}
\newcommand{\U}{{\mathbf{U}}}
\newcommand{\V}{{\mathbf{V}}}

\newcommand{\w}{{\mathbf{w}}}
\newcommand{\sg}{\text{sgn}}
\newcommand{\R}{{\mathcal{R}}}
\newcommand{\I}{{\mathcal{I}}}
\newcommand{\Rb}{{\mathbb{R}}}
\newcommand{\Cb}{{\mathbb{C}}}
\newcommand{\Sb}{{\bar{S}}}
\newcommand{\Sc}{{\mathcal{S}}}
\newcommand{\Ac}{\mathcal{A}}

\newcommand{\Y}{{\mathbf{Y}}}
\newcommand{\N}{\mathcal{N}}

\newcommand{\la}{{\lambda}}
\newcommand{\E}{{\mathbb{E}}}

\newcommand{\C}{\mathbf{C}}

\makeatletter
\def\@normalsize{\@setsize\normalsize{10pt}\xpt\@xpt
\abovedisplayskip 10pt plus2pt minus5pt\belowdisplayskip 
\abovedisplayskip \abovedisplayshortskip \z@ 
plus3pt\belowdisplayshortskip 6pt plus3pt 
minus3pt\let\@listi\@listI}
\def\subsize{\@setsize\subsize{12pt}\xipt\@xipt}
\def\section{\@startsection {section}{1}{\z@}{1.0ex plus 1ex minus .2ex}{.2ex plus .2ex}{\large\bf}}
\def\subsection{\@startsection {subsection}{2}{\z@}{.2ex plus 1ex} {.2ex plus .2ex}{\subsize\bf}}
\makeatother

\begin{document}
\date{}
\title{\Large\bf Recovering Jointly Sparse Signals via Joint Basis Pursuit}
\author{\begin{tabular}[t]{c@{\extracolsep{8em}}c} 
Samet Oymak  & Babak Hassibi \end{tabular}\\
 \\
        Department of Electrical Engineering \\
       California Institute of Technology \\
       Pasadena, CA~~91125
}
\maketitle
\thispagestyle{empty}
\subsection*{\centering Abstract}
{\em
This work considers recovery of signals that are sparse over two bases. For instance, a signal might be sparse in both time and frequency, or a matrix can be low rank and sparse simultaneously. To facilitate recovery, we consider minimizing the sum of the $\ell_1$-norms that correspond to each basis, which is a tractable convex approach. We find novel optimality conditions which indicates a gain over traditional approaches where $\ell_1$ minimization is done over only one basis. Next, we analyze these optimality conditions for the particular case of time-frequency bases. Denoting sparsity in the first and second bases by $k_1,k_2$ respectively, we show that, for a general class of signals, using this approach, one requires as small as $O(\max\{k_1,k_2\}\log\log n)$ measurements for successful recovery hence overcoming the classical requirement of $\Theta(\min\{k_1,k_2\}\log(\frac{n}{\min\{k_1,k_2\}}))$ for $\ell_1$ minimization when $k_1\approx k_2$. Extensive simulations show that, our analysis is approximately tight. 
}
\\\begin{keywords}basis pursuit, compressed sensing, phase retrieval, duality, convex optimization
\end{keywords}

\let\thefootnote\relax\footnotetext{This work was supported in part by the National Science Foundation under grants CCF-0729203, CNS-0932428 and CCF-1018927, by the Office of Naval Research under the MURI grant N00014-08-1-0747, and by Caltech's Lee Center for Advanced Networking.}

\section{Introduction}
Compressed sensing is concerned with the recovery of sparse vectors and has recently been the subject of immense interest. One of the main methods is Basis Pursuit (BP) where the $\ell_1$ norm is minimized subject to convex constraints. Assuming $\x$ has a sparse representation over the basis $\U$ (i.e.~$\U\x$ is a sparse vector) and assuming we get to see the observations $\A\x$, Basis Pursuit performs the following optimization to get back to $\x$.
\beq
\nonumber\min_{\hx}\|\U\hx\|_1~~~\text{subject to}~~~\A\x=\A\hx~~~\hspace{20pt}\text{(BP)}
\eeq
In this work, we'll be investigating recovery of vectors that can be sparsely represented over two bases. For example, a vector such as a Dirac comb can be sparse in time and frequency. Similarly, we can consider a low rank matrix which is supported over an unknown submatrix and zero elsewhere and hence sparse. Assuming $\x$ is sparse over $\U_1,\U_2$, in order to induce sparsity in both bases, we will be considering the following approach, which we call Joint Basis Pursuit (JBP).\vspace{-3pt}
\beq
\nonumber\min_{\hx}\|\U_1\hx\|_1+\la\|\U_2\hx\|_2~~~\text{s.t.}~~~\A\hx=\A\x~~~\hspace{10pt}\text{(JBP)}
\eeq
For the case of a matrix $\X$ that is simultaneously sparse and low rank, we may minimize the summation of $\ell_1$ norm and the matrix nuclear norm, which is denoted by $\|\cdot\|_\star$ and is equal to summation of the singular values. Assuming, we observe linear measurements $\Ac(\X)$, we propose solving the following problem (JBP-Matrix) to recover $\X$.\vspace{-3pt}
\beq
\nonumber \min_{\hX}\|\hX\|_\star+\la\|\hX\|_1~~~\text{s.t.}~~~\Ac(\hX)=\Ac(\X)~~~\hspace{5pt}\text{(JBPM)}
\eeq
While it is possible to come up with relevant problems, this paper will focus on JBP and JBPM. Our motivations are,
\begin{itemize}
\vspace{-1pt}
\item Investigating whether JBP can outperform regular BP.
\vspace{-1pt}
\item The sparse phase retrieval problem, in which one has measurements of a sparse vector $\x$ and observe $|\li\ab_i,\x\ri|^2$ as measurements \cite{Can4, Eldar}. While it is not possible to cast this as a regular compressed sensing problem, it can be cast as JBPM where we wish to recover sparse and low rank matrix, $\x\x^*$. This problem is known to have applications to X-Ray crystallography \cite{Mil} and has recently attracted interest \cite{Can4,Eldar, Vet,luzum}.
\end{itemize}
\vspace{0pt}
{\bf{Background:}} It should be emphasized that, recently, there has been significant interest in using a combination of different norms to exploit the structure of a signal. While this paper deals with signals having sparse representations in both bases, \cite{Can1,Can3,Cha} considers the problem of separating the signals that are combinations of sparsely representable incoherent pieces. \\
{\bf{Contributions:}} In this work we provide sharp recovery conditions that guarantees success of JBP and JBPM. Next, we cast these conditions in a \emph{dual certificate} framework to facilitate analysis. For the case of time-frequency bases, we analyze the dual certificate construction to find that for the class of ``periodic signals'', one needs at most $O(\max\{k_1,k_2\}\log\log n)$ measurements where $k_1,k_2$ represents the sparsity in $\U_1,\U_2$. This shows that JBP can indeed outperform regular BP which requires $\Theta(k\log\frac{n}{k})$ measurements for recovery of a $k$ sparse vector \cite{Ind,Don}. Finally, simulation results indicate that our results are sharp. We believe that, the result of this paper can be seen as negative in nature. While, JBP provides an improvement, it is not a significant improvement when we consider the fact that signals that are simultaneously sparse are \emph{few} in number. 
%


\section{Problem Setup}
We begin by considering the (JBP) problem and assume $\x\in{\Cb}^{n}$ is a signal that is sparse over two complete bases, $\U_1,\U_2$. Later on we will briefly extend our approach to (JBPM) and the recovery of matrices that are simultaneously sparse and low rank. \\
\indent The basic question we would like to answer is whether one can do better in recovering $\x$ from measurements $\A\x$ by exploiting the joint sparsity of $\x$.\\
\indent Before, going into technical details, we'll introduce the relevant notation.
Denote the set $\{1,2,\dots,n\}$ by $[n]$. Let $S_1,S_2\subseteq[n]$ denote the supports of $\x$ in the bases $\U_1$ and $\U_2$, i.e., locations of nonzero entries of $\U_1\x$ and $\U_2\x$ respectively. Further, let $\Sc_1(\cdot):\Cb^n\rightarrow\Cb^{|S_1|},\Sc_2(\cdot):\Cb^n\rightarrow\Cb^{|S_2|}$ denote the operators that collapse a vector onto $S_1,S_2$ respectively. $\sg(\cdot):\Cb^n\rightarrow\Cb^n$ is the function that returns entry wise signs of a vector, i.e., $0$ is mapped to $0$ and $a\neq 0$ is mapped to $\frac{a}{|a|}$. $\Ib$ will be the identity matrix of the appropriate size. Null space of a linear operator $\A$ is denoted by $\N(\A)$. $\R(\cdot),\I(\cdot):\Cb^n\rightarrow\Rb^n$ are the functions that returns entry-wise real and imaginary parts of a vector. Denote $\sqrt{-1}$ by $\ib$. $\D$ is the Discrete Fourier Transform (DFT) matrix of the appropriate size and given as follows,
\vspace{-5pt}
\beq
\D_{i,j}=\frac{W^{(i-1)(j-1)}}{\sqrt{n}}~~~~1\leq i,j\leq n
\eeq
where $W$ is always $\exp(-\frac{2\pi \ib}{n})$. We will use $\la_1,\la_2$ and $1,\la$ alternatively.

\noindent {\bf{Remark:}} Proofs that are omitted can be found in the appendix.

\subsection{Recovery Conditions for JBP}
We will start with explaining our approach. Let $\A\in\Cb^{m\times n}$ where $m$ is the number of measurements. The following lemma gives a condition that guarantees $\x$ to be the unique optimum of (JBP).
\begin{lem}[Null Space Condition] Assume, for all $\w\in\N(\A)$, the following holds,
\beq
\label{null}
\sum_{i=1}^2\la_i(\R(\li\sg(\U_i\x),\U_i\w\ri)+|\Scb_i(\U_i\w)|_1)>0
\eeq
Then, $\x$ is the unique optimizer of (JBP).
\end{lem}
\begin{proof}
Let $f(\hx)$ be the cost of (JBP) , i.e., $f(\hx)=\sum_{i=1}^2\la_i\|\U_i\hx\|_1$. Then, for any $\w\in\N(\A)$, $f(\x+\w)-f(\x)$ is lower bounded by the left hand side of (\ref{null}), which follows from the sub gradient of the $\ell_1$ norm. Hence $f(\hx)>f(\x)$ for all $\A\hx=\A\x$, $\hx\neq\x$.
\end{proof}
Based on (\ref{null}), the following lemma connects success of (JBP) to the existence of dual certificates.
\begin{lem}
\label{lemdual}
 Assume $\s_1,\s_2\in\Cb^m,\s\in\Cb^n$ satisfying the following conditions exist:
\begin{itemize}
\item $\Sc_1(\U_1^{-*}(\A^*\s_1+\s))=\Sc_1(\sg(\U_1\x))$
\item $\|\Scb_1(\U_1^{-*}(\A^*\s_1+\s))\|_\infty<1$
\item $\Sc_2(\U_2^{-*}(\A^*\s_2-\s))=\la\Sc_2(\sg(\U_2\x))$
\item $\|\Scb_2(\U_2^{-*}(\A^*\s_2-\s))\|_\infty<\la$
\item $\A$ is invertible over $\bigcap_{i=1}^2\{\vv\big|\Sc_i(\U_i\vv)=\U_i\vv\}$.
\end{itemize} 
Then $\x$ is the unique optimum of (JBP).
\end{lem}
\begin{proof}What we need to show is that if such $\s_1,\s_2,\s$ exist and the invertibility assumption holds then the left hand side of (\ref{null}) is strictly positive for all $\w\in\N(\A)$.
Assume such $\s_1,\s_2,\s$ exist and let $\vv_1,\vv_2\in\Cb^{n}$ to be:
\vspace{-5pt}
\begin{equation}
\vv_1=\U_1^{-*}(\A^*\s_1+\s)~~~\text{and}~~~\vv_2=\U_2^{-*}(\A^*\s_2-\s^*)
\end{equation}
Observe that for any $\w\in\N(\A)$, using $\A\w=0$,
\vspace{-5pt}
\begin{equation}
\label{eqzero}
\sum_{i=1}^2\li\U_i\w,\vv_i\ri=\sum_{i=1}^2\w^*\A^*\s_i+\w^*\s-\w^*\s=0
\end{equation}
To end the proof observe that $\vv_1,\vv_2$ satisfies the conditions listed in Lemma \ref{lemdual} which implies that the LHS of (\ref{null}) is strictly positive when combined with (\ref{eqzero}). This follows from the fact that either $\Scb_1(\U_1\w)$ or $\Scb_2(\U_2\w)$ is nonzero due to invertibility assumption.
\end{proof}
The dual certificate approach for regular BP has been used in \cite{Tro,Can2,Can3}. Letting $\U=\U_1$, compared to Lemma \ref{lemdual}, it requires invertibility of $\A$ over $\{\vv\big|\Sc_1(\U_1\vv)=\U_1\vv\}$ rather than the intersection and it requires $\|\Scb_1(\U_1^{-*}\A^*\s_1)\|_\infty<1$, while Lemma \ref{lemdual} can overcome this by making use of the extra variable $\s$. From this perspective, JBP can be viewed as a combination of two regular BP's that are allowed to ``help'' each other via $\s$.

\section{Main Results}
Our main result is concerned with the time-frequency bases, i.e., Identity and the DFT matrices. Before stating the main result, let us first describe the setting for which it holds.
\begin{defn}$S$ is a $l$ periodic subset of $[n]$ if $n$ is divisible by $l$ and for any $i\in[n]$, we have,
\beq
i\in S\iff j\in S~\text{for all $j$ such that}~j\equiv i~(\text{mod}~l)
\eeq
Observe that if $S$ is a $l$ periodic support, $|S|$ is divisible by $n/l$.
\end{defn}
\begin{thm}\label{main}
Let $\U_1=\Ib$, $\U_2=\D$, $1> \alpha\geq 0$ be an arbitrary constant and without loss of generality assume $|S_1|\leq |S_2|$. Further, assume the followings hold,
\begin{itemize}
\item $|S_1|\leq \frac{n}{\log n}$.
\item $S_1,S_2$ are $n_1,n_2$ periodic supports, where $n=n_1n_2$.
\item $|S_2|\leq |S_1|\log^{\alpha}(n)$.
\end{itemize}
Then, for the following scenarios, $\x$ can be successfully recovered via JBP with high probability (for sufficiently large $n$) when the matrix $\A\in\Cb^{m\times n}$ is generated with i.i.d complex Gaussian entries.
\begin{itemize}
\item If $ |S_2|\leq |S_1|\log\log n$ setting $\la=1$ and using $m=O(|S_2|\log\log n)$ measurements.
\item If $|S_2|\geq |S_1|\log\log n$, setting $\la=\log^{-1}(n)$ and using $m=O(|S_2|)$ measurements.
\end{itemize}
\end{thm}
{\bf{Remark:}} Our proof approach will inherently require $m\geq \max\{|S_1|,|S_2|\}$. Consequently, if $|S_2|\geq |S_1|\log(n)$, then one can already perform the regular $\ell_1$ optimization over $\U_1=\Ib$ to ensure recovery with $m=O(|S_2|)$ measurements. Hence, $|S_2|\leq |S_1|\log^{\alpha}(n)$ is a reasonable assumption.

\subsection{Signals with Periodic Supports}
Theorem \ref{main} holds for signals whose supports are periodic with $n_1,n_2$ over $\Ib$ and $\D$ respectively, where $n=n_1n_2$. Here, we give a family of such signals that satisfy this requirement. Let $T$ be the set of signals $\vv\in\Cb^n$ such that for some $l\leq n_1$ and $0\leq t<n$,
\beq
v_j=\begin{cases}0~\text{if}~j\not\equiv l~\text{(mod $n_1$)}\\W^{ jt}~\text{else}\end{cases}
\eeq
Basically, $T$ is the set of Dirac combs with period $n_1$ and hence for any $\vv\in T$, $\D\vv$ will have $\frac{n}{n_1}$ periodic support. In general, almost all $\x$ of the form,
\beq
\x=\sum_{\vv_i\in T}\alpha_i\vv_i
\eeq
will have $n_1$ periodic support and $\D\x$ will have $\frac{n}{n_1}$ periodic support. The reason we say almost all is because cancellations may occur when $\vv_i$'s are added. However, if $\alpha_j$'s are chosen from a continuous distribution, the chance of cancellation is $0$.

\subsection{Converse Results}
We should emphasize that, the main reason we have considered the $\Ib,\D$ pair is the fact that almost all bases $\U_1$ and $\U_2$ do not permit signals that are sparse in both. The following lemma illustrates this.
\begin{lem}\label{triv}Assume $\U_1^{-1},\U_2^{-1}$ have i.i.d entries chosen from a continuous distribution. Then, with probability $1$, there exists no nonzero vector $\x$ satisfying $|S_1|+|S_2|\leq n$.
\end{lem}
An interesting work by Tao shows that, such results are true even for highly structured bases, \cite{Tao}. In particular, if $n$ is a prime number, we still have $|S_1|+|S_2|>n$ requirement for a signal over $\U_1=\Ib$ and $\U_2=\D$ bases.

\section{Proof of Theorem \ref{main}}
This section will be dedicated to the analysis of Lemma \ref{lemdual} to prove Theorem \ref{main}. We start by proposing a construction for $\s_1,\s_2,\s$ that certifies optimality of $\x$.

%
%
%
%
%
%
\subsection{Construction of $\s_1,\s_2,\s$}
For the following discussion, we'll be using $(\U_1,\U_2)$ and $(\Ib,\D)$ and $(1,\la)$ and $(\la_1,\la_2)$ interchangeably. The construction of $\s_1,\s_2$ will follow a classical approach previously used in \cite{Can3, Tro, Can4}. Letting $\A_{S_1}\in\Cb^{m\times |S_1|}$ denote the submatrix by choosing columns corresponding to $S_1$ and $\B=\A\D^*$, we will use the following $\s_1,\s_2$.
\bea
&\s_1=\A_{S_1}(\A_{S_1}^*\A_{S_1})^{-1}\Sc_1(\sg(\x))\\
&\s_2=\B_{S_2}(\B_{S_2}^*\B_{S_2})^{-1}\la\Sc_2(\sg(\D\x))
\end{align}
Since $\Ib,\D$ are unitary we have $\U_i^{-*}=\U_i$. By construction $\s_1,\s_2$ already satisfies,
\beq
\Sc_i(\U_i\A^*\s_i)=\la_i\Sc_i(\sg(\U_i\x))~~~i\in\{1,2\}
\eeq
However, one has to control the term $\|\Scb_i(\U_i\A^*\s_i)\|_\infty$ and we will make use of $\s$ to achieve this. Denote $\U_i\A^*\s_i$ by $\y_i$. Define the vectors $\{\bb_1,\bb_2\}$ as follows:
\beq\label{myb}
\R(b_{i,j})=\begin{cases}0~\text{if}~j\in S_i\\0~\text{if}~j\in\Sb_i~\text{and}~|\R(y_{i,j})|\leq \la_i/4\\\R(y_{i,j})-\la_i\sg(\R(y_{i,j}))/4~\text{else}\end{cases}
\eeq
and imaginary part $\I(b_{i,j})$ is obtained from $\I(y_{i,j})$ in the same way.
Observe that, $\|\Scb_i(\y_i-\bb_i)\|_\infty< \la_i/2$. Based on $\{\bb_i\}_{i=1}^2$ construct $\s$ as follows,
\begin{align}
&\s=\D^*(\bb_2-\cb_2)-\Ib(\bb_1-\cb_1)~~~\text{where}~\\
\nonumber&\cb_1=\D^*\Ib_{S_2}\D\bb_1,~\cb_2=\D\Ib_{S_1}\D^*\bb_2
\end{align}
Here, $\Ib_{S_1}$, $\Ib_{S_2}$ are diagonal matrices whose diagonal entries corresponding to $S_1,S_2$ are $1$ and the rest are zero.

\begin{lem}\label{nice} Assume $\x,\{\y_i,\bb_i,\cb_i\}_{i=1}^2$ are the same as described previously. Then, one has the following:
\begin{align}
\nonumber&\Sc_1(\y_1+\s)=\Sc_1(\sg(\x))\\
\nonumber&\Sc_2(\y_2-\D\s)=\la\Sc_2(\sg(\D\x))\\
\nonumber&\|\Scb_1(\y_1+\s)\|_\infty< \frac{1}{2}+\|\Scb_1(\cb_1)\|_\infty+\|\Scb_1(\D^*\bb_2)\|_\infty\\
\nonumber&\|\Scb_2(\y_2-\D\s)\|_\infty<\frac{\la}{2}+\|\Scb_2(\cb_2)\|_\infty+\|\Scb_2(\D\bb_1)\|_\infty
\end{align}
\end{lem}
Based on Lemma \ref{nice} and Lemma \ref{lemdual}, JBP recovers $\x$ if we have, $\|\Scb_1(\cb_1)\|_\infty+\|\Scb_1(\D^*\bb_2)\|_\infty\leq 1/2$ and $\|\Scb_2(\cb_2)\|_\infty+\|\Scb_2(\D\bb_1)\|_\infty\leq\la/2$.

As a next step, we can analyze $\|\Scb_1(\D^*\Ib_{S_2}\D\bb_1)\|_\infty$ and $\|\Scb_1(\D^*\bb_2)\|_\infty$ and find the conditions that guarantees their sum to be small. The analysis for $S_2$ will be identical to $S_1$ and hence is omitted.

%
%
%
%
%
\subsection{Probabilistic Analysis}
Assume $\A$ is i.i.d complex normal with variance $\frac{1}{m}$ and $m\geq 64\max\{|S_1|,|S_2|\}$. This will guarantee,
\beq\label{goodcon}
\sigma_{min}(\A_{S_1})\geq 1/\sqrt{2}~\text{and}~\sigma_{min}(\B_{S_2})\geq 1/\sqrt{2}
\eeq
with probability $1-\exp(-\Omega(m))$, \cite{Ver}.

Now, conditioned on $\A_{S_1},\B_{S_2}$ satisfy (\ref{goodcon}),
\beq
\nonumber\|\s_i\|_2^2=\s_i^*\s_i=\la_i^2\sg(\x)^*(\A_{S_1}^*\A_{S_1})^{-1}\sg(\x)\leq 2\la_i^2|S_i|
\eeq
and $\Scb_i(\y_i)$ is an i.i.d Gaussian vector whose entries having variance $\frac{\|\s_i\|_2^2}{m}$. Given these, we need to understand, when can we make sure,
\beq
\nonumber\|\Scb_1(\D^*\Ib_{S_2}\D\bb_1)\|_\infty\leq \frac{1}{4}~\text{and}~\|\Scb_1(\D^*\bb_2)\|_\infty\leq \frac{1}{4}
\eeq
From (\ref{myb}), observe that $\Scb_i(\bb_i)$ is a function of $\Scb_i(\y_i)$ which is i.i.d. random Gaussian. The next lemma, gives a characterization of $\bb_i$.
\begin{lem}\label{bob} Assume $m\geq 64\max\{|S_1|,|S_2|\}$. Then, the entries $\{\Scb_i(\bb_i)_j\}_{j=1}^{|\Sb_i|}$ of $\Scb_i(\bb_i)$ are i.i.d. random variables with the following distribution,
\beq
\Scb_i(\bb_i)_j~\text{is}~\begin{cases}0~\text{with probability at least}~1-4\exp(-\frac{m}{16|S_i|})\\\text{otherwise distributed as}~\z\end{cases}
\eeq
where $\z$ is $0$ mean and subgaussian norms (see \cite{Ver}) of $\R(\z),\I(\z)$ are upper bounded by $c_0\la_i\sqrt{\frac{|S_i|}{m}}$ for an absolute constant $c_0>0$.
\end{lem}
\subsubsection{Analysis of $\|\Scb_1(\cb_1)\|_\infty$}
We need to show,
\beq\label{show1}
\|\Scb_1(\D^*\Ib_{S_2}\D\bb_1)\|_\infty\leq \frac{1}{4}
\eeq
Calling $\C=\D^*\Ib_{S_2}\D$, from Lemma \ref{useful}, each row of $\C$ has energy $\frac{|S_2|}{n}$. Let $\cb_i$ be the $i$'th column of $\C^*$. Then, using Lemma \ref{bob} and Proposition $5.10$ of \cite{Ver}, for any $i$ and an absolute constant $c>0$,
\begin{align}\label{show2}
\Pro(|\cb_i^*\bb_1|\geq \frac{1}{4})&\leq 12\exp(-\frac{mc}{2^7c_0^2|S_1|\|\cb_{i,T}\|_2^2})\\
&=12\exp(-\frac{mnc}{2^7c_0^2|S_1||S_2|})
\end{align}
Using a union bound over all $i$'s, shows (\ref{show1}) reduces to arguing $n\Pro(|\cb_i^*\bb_1|\geq \frac{1}{4})\rightarrow 0$ which is equivalent to ensuring,
\beq\label{show3}
\frac{mnc}{2^7c_0^2|S_1||S_2|}- \log n\rightarrow \infty~~~\text{as}~~~n\rightarrow\infty
\eeq
Using $n\geq \min\{|S_1|,|S_2|\}\log n$ in the statement of Theorem \ref{main}, (\ref{show3}) holds for $m\geq 2^8c^{-1}c_0^2\max\{|S_1|,|S_2|\}=O(\max\{|S_1|,|S_2|\})$ as desired. 
\subsubsection{Analysis of $\|\Scb_1(\D^*\bb_2)\|_\infty$}
In a similar fashion, we would like to show,
\beq\label{wanted}
\|\Scb_1(\D^*\bb_2)\|_\infty\leq \frac{1}{4}
\eeq
holds with high probability, to conclude. Each row of $\D^*$ has unit $\ell_2$ norm and nonzero entries of $\bb_2$ are i.i.d subgaussians from Lemma \ref{bob}. Letting, $p=4\exp(-\frac{m}{16|S_2|})$ and applying a Chernoff bound w.p.a.l $1-\exp(-np/4)$, number of non zeros in $\bb_2$ is at most $2np$. Considering the inner products between each row of $\D^*$ and $\bb_2$, and using a union bound, (\ref{wanted}) holds, with probability at least,
\beq\label{second}
1-12n\exp(-\frac{mc}{2^8c_0^2\la^2|S_2|p})-\exp(-np/4)
\eeq
Assuming $m=O(|S_2|\log^\alpha (n))$ for some $\alpha<1$, we have $\exp(-np/4)\rightarrow 0$. Finally, to show the second term in (\ref{second}) approaches $0$, for some absolute constants $c_1,c_2>0$, we need to argue,
\beq
\frac{m}{c_1\la^2|S_2|}\exp(\frac{m}{c_2|S_2|})-\log n\rightarrow\infty~~~\text{as}~~~n\rightarrow\infty
\eeq
Following the same arguments for the other basis will yield,
\beq
\frac{m\la^2}{c_1|S_1|}\exp(\frac{m}{c_2|S_1|})-\log n\rightarrow\infty~~~\text{as}~~~n\rightarrow\infty
\eeq
By choosing $m=O(\max\{|S_1|,|S_2|\}\log\log n)$ and $\la=1$ one can always satisfy these. In case $|S_2|\geq |S_1|\log\log n$, choose $\la=\log^{-1}(n)$ and $m$ sufficiently large but $O(|S_2|)$ to still satisfy both.

\section{Empirical Results}

%
%
%
%

While Theorem \ref{main} shows that JBP can indeed outperform BP it is important to understand how good it actually is. We considered the following basic setup: Let $k$ be a positive integer and $n=k^2$. Then, let $\x\in\Rb^n$ be the following dirac comb,
\beq
x_i=1~\text{if}~i\equiv 1~(\text{mod}~k)~\text{and}~0~\text{else}
\eeq
It is clear that $\D\x=\x$ hence the signal is only $\sqrt{n}$ sparse in both domain and the optimal weight in JBP is $\la=1$ by symmetry. Simulation for JBP is performed for $k=\{2,4,6,\dots,32\}$ and for $1\leq m\leq 30$. Interestingly, in order to achieve $50\%$ success, JBP required $\frac{k}{2}\leq m\leq \frac{3k}{5}$ and $\frac{m}{k}$ slightly increased as a function of $k$. This is shown as the straight line in Figure $1$. These results are quite consistent with Theorem \ref{main} from which we expect to have $m=O(k\log\log k)$ measurements. 

On the other hand, $50\%$ success curve for BP is shown as the dashed line in Figure $1$ and obeys $m=O(k\log k)$ as expected from classical results on $\ell_1$ minimization. In particular $\frac{m}{k}$ increases from $1$ to $2.4$ as $k$ moves from $2$ to $32$.
While JBP outperforms BP in this setting, the fact that it requires $\Omega(k)$ samples to recover a highly structured signal is disappointing. It would be interesting to see whether a greedy algorithm can be developed to attack this problem.
\begin{figure}
\label{fig2}
  \begin{center}
{\includegraphics[scale=0.26]{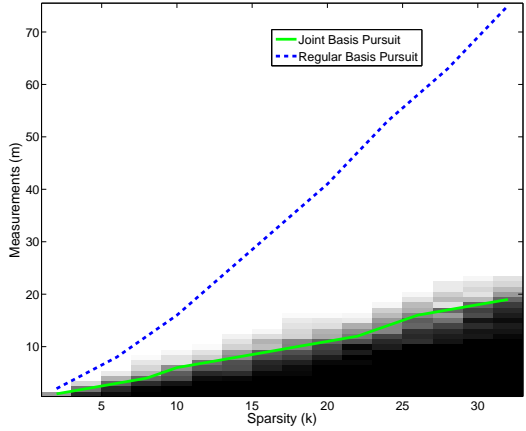}}  
  \end{center}
  \caption{Phase transitions of JBP vs BP where sparsity $k$ varies between $2$ to $32$ and $n=k^2$. Dark region indicates failure for JBP while light region corresponds success. Straight and dashed lines are $50\%$ success curves for JBP and BP respectively. While JBP outperforms BP, it still requires $\Omega(k)$ measurements.}
  \vspace*{-10pt}
\end{figure}

\section{Extension to Matrices}
As it has been discussed in the introduction, similar to jointly sparse signals one might as well consider matrices that are sparse and low rank. The motivation is the sparse phase retrieval problem where $\x$ is a sparse vector to be recovered from observations $\{|\li\ab_i,\x\ri|^2\}_{i=1}^m$ where $\{\ab_i\}_{i=1}^m\in\Cb^n$ are the measurement vectors. Although, these measurements are not linear in $\x$, they are linear in $\x\x^*$ as $|\li\ab_i,\x\ri|^2=\ab_i^*\x\x^*\ab_i$. Using the fact that $\x\x^*$ is rank $1$ and sparse, JBPM can be used in order to recover $\X=\x\x^*$ as it will enforce a low-rank and sparse solution.

Although, this work will not deal with the analysis of this problem, we'll point out that our framework for JBP can be used for JBPM as well. In general, assume matrix $\X$ is low-rank and sparse and we wish to recover it from observations $\Ac(\X)$. Let us first introduce notation relevant to structure of $\X\in\Cb^{n\times n}$.
\begin{itemize}
\item Let $S\in[n]\times [n]$ be the usual support of $\X$ and $\Sc:\Cb^{n\times n}\rightarrow\Cb^{|S|}$ be the projection onto $S$.
\item Assuming $\X$ has singular value decomposition $\U\Sigma\V^*$, Define the subspace $\Lc\in\Cb^{n\times n}$ as,
\beq
\nonumber\Lc=\{\Y\in\Cb^{n\times n}~\big|~(\Ib-\U\U^*)\Y(\Ib-\V\V^*)=0\}
\eeq
\item $\bar{\Lc}$ denotes complement of $\Lc$ and projection onto $\Lc$ is denoted by $\Lc(\cdot):\Cb^{n\times n}\rightarrow\Cb^{n\times n}$. 
\item $\Ac^{*}(\cdot):\Cb^m\rightarrow\Cb^{n\times n}$ denotes the adjoint operator. Operator norm is denoted by $\|\cdot\|$.
\end{itemize}

The following lemma is effectively equivalent to Lemma \ref{lemdual} and characterizes a simple condition for $\X$ to be unique optimizer of JBPM.
\begin{lem}
\label{lemdualmat}
Assume $\SB_1,\SB_2\in\Cb^m,\SB\in\Cb^{n\times n}$ satisfying the following conditions exist:
\begin{itemize}
\item $\Lc(\Ac^*(\SB_1)+\SB)=\U\V^*$.
\item $\|\Lcb(\Ac^*(\SB_1)+\SB)\|<1$.
\item $\Sc(\Ac^*(\SB_2)-\SB)=\lambda\cdot\Sc(\sg(\A))$.
\item $\|\Scb(\Ac^*(\SB_2)-\SB)\|_\infty<\lambda$.
\item $\Ac(\cdot)$ is invertible over $\{\Y\big|\Lc(\Y)=\Sc(\Y)=\Y\}$.
\end{itemize} 
Then $\A$ is the unique optimum of (JBPM).
\end{lem}
Finally, it would be interesting to see whether similar or better improvements can be shown for JBPM over regular BP or regular nuclear norm minimization algorithms.

\pagebreak
\section{Appendix}

We will start by proving Lemma \ref{triv} using a classical argument.
\begin{proof}[\bf{Proof of Lemma \ref{triv}}] Let us first fix $S_1,S_2$ and consider these particular supports. Let $\C_i\in\R^{n\times |S_i|}$ be the matrix obtained by taking columns of $\U_i^{-1},$ over $S_i$. If $\z_1=\U_1\x$ and $\z_2=\U_2\x$ are supported over $S_1,S_2$, we may write:
\beq\label{z12}
0=\U_1^{-1}\z_1-\U_2^{-1}\z_2=[\C_1~\C_2][\Sc_1(\z_1)^*~-\Sc_2(\z_2)^*]^*
\eeq
By assumption, $[\C_1~\C_2]\in\R^{n\times (|S_1|+|S_2|)}$ has i.i.d. entries from a continuous distribution and hence full column rank with probability $1$ whenever $|S_1|+|S_2|\leq n$. It follows that only $(\z_1,\z_2)$ satisfying (\ref{z12}) is $(0,0)$. There are finitely many $S_1,S_2$ pairs satisfying $|S_1|+|S_2|\leq n$ hence a union bound will still give, with probability $1$, there exists no nonzero vector $\x$ having combined sparsities of $\U_1\x$ and $\U_2\x$ at most $n$.
\end{proof}
\vspace{10pt}
Following lemma gives a simple but useful property of the DFT matrix.

\begin{lem}\label{useful} Let $n=n_1n_2$ and $S_1,S_2$ be $n_1$ and $n_2$ periodic supports. Let $\D\in\Rb^{n\times n}$ be the DFT matrix as previously. Further, let $\C=\D^*\Ib_{S_2}\D$. Then,
\begin{enumerate}
\item $\C_{i,j}=0$ for any $(i,j)$ with $i\not\equiv j~(\text{mod}~n_1)$.
\item For any $i$, $i$'th row $\rb_i$ of $\C$ satisfies $\|\cb_i\|_2^2=\frac{|S_2|}{n}$.
\item For any $\x$ that is supported on $\Sb_1$, we have, $\Sc_1(\C\x)=0$.
\item First three results similarly hold for $\C=\D\Ib_{S_2}\D^*$.
\end{enumerate}
\end{lem}
\begin{proof}
Let us start by analyzing the matrix $\D^*\Ib_{S_2}\D$. Let $\db_i$ be the $i$'th column of $\D$. Then, 
\beq
\C_{i,j}=\db_i^*\Ib_{S_2}\db_j=\sum_{k\in S_2}\db_{i,k}^*\db_{j,k}
\eeq
Using $S_2$ is $n_2$ periodic, for some set $T\in[n_1]$ (which is simply $S_2~(\text{mod}~n_1)$), we may write,
\beq
\C_{i,j}=\sum_{t\in T}\sum_{c=1}^{n_1}\cb_{i,t+cn_2}^*\cb_{j,t+cn_2}
\eeq
Next, for any $i\not\equiv j~(\text{mod}~n_1)$ and any $t\leq n_1$,
\begin{align*}
\sum_{c=0}^{n_1-1}\cb_{i,t+cn_2}^*\cb_{j,t+cn_2}&=\sum_{c=0}^{n_1-1} W^{(j-i)(t+cn_2)}\\
&=W^{(j-i)t}\sum_{c=0}^{n_1-1} W^{c(j-i)n_2}\\
&=n_1W^{(j-i)t}\delta(i-j~(\text{mod}~n_1))
\end{align*}
where $\delta(k)=1\iff k\neq 0$. This proves the first statement. To show the second, $\rb_i=\db_i^*\Ib_{S_2}\D$ implies:
\begin{align}
\nonumber \|\rb_i\|_2^2=\rb_i^*\rb_i=\db_i^*\Ib_{S_2}\D\D^*\Ib_{S_2}\db_i=\db_i^*\Ib_{S_2}\db_i=\frac{|S_2|}{n}
\end{align}
Third result will be a direct consequence of the first one: If $\x\in\Sb_1$, then
\beq
(\C\x)_i=\sum_{j=1}^n\C_{i,j}\x_j=\sum_{j\in \Sb_1}\C_{i,j}\x_j
\eeq
When $i\in S_1,j\in \Sb_1$, we have $i\not\equiv j(\text{mod}~n_1)$ by definition, which implies $\C_{i,j}=\C_{i,j}\x_j=0$ due to the first result. Fourth result can be shown by repeating these arguments for $\D\Ib_{S_2}\D^*$.
\end{proof}

Using Lemma \ref{useful}, we'll now proceed with the proof of Lemma \ref{nice}.

\subsection{Proof of Lemma \ref{nice} }
\begin{proof}$S_1$ and $S_2$ components will be analyzed seperately.\\
{\bf{Analyzing $S_1$}}: We may start by considering, $\y_1+\s$ and write,
\begin{align*}
\y_1+\s&=\y_1+\D^*(\bb_2-\cb_2)-(\bb_1-\cb_1)\\
&=\y_1+\D^*\bb_2-\Ib_{S_1}\D^*\bb_2-\bb_1+\D^*\Ib_{S_2}\D\bb_1
\end{align*}
First, we'll consider, $\Sc_1(\y_1+\s)$. We have the following,
\begin{align}
&\Sc_1(\y_1)=\Sc_1(\sg(\x))~\text{by construction of}~\y_1\\
&\Sc_1(\D^*\bb_2-\Ib_{S_1}\D^*\bb_2)=\Sc_1((\Ib-\Ib_{S_1})\D^*\bb_2)=0\\
&\Sc_1(\bb_1)=0~\text{by construction of}~\bb_1\\
&\Sc_1(\D^*\Ib_{S_2}\D\bb_1)=0~\text{from Lemma \ref{useful}.}
\end{align}
Hence, we find, $\Sc_1(\y_1+\s)=\Sc_1(\y_1)=\Sc_1(\sg(\x))$.\\
To upper bound $\|\Scb_1(\y_1+\s)\|_\infty$, we may simply use $\|\Scb_1(\y_1-\bb_1)\|_\infty<1/2$ and write,
\beq
\nonumber\|\Scb_1(\y_1+\s)\|_\infty\leq \|\Scb_1(\y_1-\bb_1)\|_\infty+\|\Scb_1(\cb_1)\|_\infty+\|\Scb_1(\D^*\bb_2)\|_\infty
\eeq

\noindent{\bf{Analyzing $S_2$}}: Similarly, for $\Sc_2(\y_2+\D\s)$, we have the following,
\begin{align}
&\Sc_2(\y_2)=\la\Sc_2(\sg(\D\x))~\text{by construction}\\
&\Sc_2(\D\bb_1-\Ib_{S_2}\D\bb_1)=\Sc_2((\Ib-\Ib_{S_2})\D\bb_1)=0\\
&\Sc_2(\bb_2)=0~\text{by construction}\\
&\Sc_2(\D\Ib_{S_1}\D^*\bb_2)=0~\text{from Lemma \ref{useful}.}
\end{align}
Hence, $\Sc_2(\y_2-\D\s)=\la\Sc_2(\sg(\D\x))$ as desired.\\
To upper bound $\|\Scb_2(\y_2-\D\s)\|_\infty$, we may use $ \|\Scb_2(\y_2-\bb_2)\|_\infty<\la/2$ and write,
\beq
\nonumber\|\Scb_2(\y_2-\D\s)\|_\infty<\frac{\la}{2}+\|\Scb_2(\cb_2)\|_\infty+\|\Scb_2(\D\bb_1)\|_\infty
\eeq
\end{proof}
\subsection{Proof of Lemma \ref{bob}}
\begin{proof}
We start by stating a useful lemma on Gaussian variables, \cite{Ver}.
\begin{lem}\label{very simp}Let $g$ be a real standard normal random variable. Then, for any $t\geq 0$
\beq
\Pro(|g|>t)\leq 2\exp(-t^2/2)
\eeq
\end{lem}
Our discussion will be for $S_1$ only. Proof for $S_2$ is identical.\\
{\bf{Case 1: Estimating $\Pro(\Scb_1(b_1)_i=0)$}} \\
Observe that $\A_{S_1}$ and $\A_{\Sb_1}$ are independent matrices with i.i.d. Gaussian entries. Hence, for fixed $\A_{S_1}$, $\A_{\Sb_1}$ is i.i.d. $\Scb_1(\y_1)$ is a vector with i.i.d. complex Gaussian entries with variance $\frac{\|\s_1\|_2^2}{m}$. Next, from (\ref{myb}) it can be seen that $\Scb_1(\bb_1)$ is an entry wise function of $\Scb_1(\y_1)$ and hence i.i.d. Using Lemma \ref{very simp} and conditioned on $\sigma_{min}(\A_{S_1})\geq 1/\sqrt{2}$ for any $i\in\Sb_1$
\beq
\Pro(\R(b_{1,i})=0)=\Pro(|\R(y_{1,i})|<\frac{1}{4})\geq 1-2\exp(-\frac{m}{16|S_1|})
\eeq
as variance of $\R(y_{1,i})$ is at most $\frac{|S_1|}{m}$. Using a union bound over real and imaginary parts of $b_{i,j}$, we find,
\beq
\Pro(\R(b_{1,i})=0)\geq 1-4\exp(-\frac{m}{16|S_1|})
\eeq
{\bf{Case 2: Subgaussian norm when $\Scb_1(\bb_1)_j\neq 0$}}\\
Let us first define a subgaussian random variable and its norm.
\begin{defn}Let $\z\in\Rb$ be a scalar random variable. Assume for some $K<\infty$,
\beq\label{hizli}
(\E[|\z|^n])^{1/n}\leq K\sqrt{n}~~\text{for all integers}~n\geq 1
\eeq
Then, $\z$ is a subgaussian random variable and smallest $K$ satisfying (\ref{hizli}) is norm of $\z$.
\end{defn}
Assume $i\in\Sb_1$. This time, we consider the case where $|y_{1,i}|>0$. Clearly real and imaginary components of $b_{1,i}$ are independent as it is the case for $y_{1,i}$. Without loss of generality consider the real part.
Observe that, if $\R(b_{1,i})\neq 0$ then it is $\R(y_{1,i})-\frac{1}{4}\sg(\R(y_{1,i}))$ where $\text{var}(\R(y_{1,i}))\leq \frac{|S_1|}{m}\leq \frac{1}{64}$ by assumption. Hence, using following lemma we can conclude that subgaussian norm of $b_{1,j}$ is upper bounded by $c_0\sqrt{|S_1|/m}$ as $1/4>\sqrt{2}/8$.
\begin{lem} \label{sub gauss}Let $c\geq \sqrt{2}$ be a scalar, $x$ be a standard normal random variable and,
\beq
z=x-c\cdot\sg(x)~\text{conditioned on}~|x|\geq c
\eeq
Then, $z$ has subgaussian norm at most $c_0$ for some absolute constant $c_0$.
\end{lem}
\begin{proof}
Following inequality is true for tail of Gaussian p.d.f,
\beq
\nonumber\frac{1}{\sqrt{2\pi}x}(1-\frac{1}{x^2})\exp(-x^2/2)<Q(x)< \frac{1}{\sqrt{2\pi}x}\exp(-x^2/2)
\eeq
Hence, using $c\geq \sqrt{2}$, for $t\geq 0$ we have,
\begin{align*}
\Pro(|z|>t)&=\frac{Q(t+c)}{Q(c)}\leq \frac{c\exp(-(t+c)^2/2)}{(t+c)(1-c^{-2})\exp(-c^2/2)}\\
&\leq 2\exp(-t^2/2)
\end{align*}
Result immediately follows from Lemma 5.5 of \cite{Ver} and from the bound on $\Pro(|z|>t)$.
\end{proof}
Finally, $b_{1,i}$ is zero mean as $y_{1,i}$ is distributed symmetrically around $0$ and construction of $b_{1,i}$ preserves the symmetry.
\end{proof}
\subsection{Proposition $5.10$ and sums of sub-gaussians}
Next, we state Proposition $5,10$ of \cite{Ver} for completeness, which gives a bound on weighted sum of subgaussians.
\begin{thm}[Proposition $5.10$ of \cite{Ver}]Let $\z_1,\dots,\z_l$ be subgaussian random variables with subgaussian norms upper bounded by $c_0>0$. Let $\ab\in\R^l$ be an arbitrarily chosen vector. Then, for all $t\geq 0$,
\beq
\Pro(|\sum_{i=1}^l a_i\z_i|\geq t)\leq 3\exp(-\frac{ct^2}{c^2\|\ab\|_2^2})
\eeq
where $c>0$ is an absolute constant.
\end{thm}
Based on this, we can obtain (\ref{show2}) as $\R(\Scb_1(\bb_1))$ is i.i.d. subgaussian with norm at most $c_0\sqrt{|S_1|/m}$ and we need to argue both contributions from real and imaginary parts are at most $\frac{1}{4\sqrt{2}}$ with high probability. In particular for $j$'th row of $\C$,
\begin{align}
\Pro( |\sum_i \R(c_{j,i})\R(b_{1,i})|>\frac{1}{8\sqrt{2}})\leq 3\exp(-\frac{cmn}{128|S_1||S_2|})
\end{align} 
Writing similar bounds for $|\sum_i \I(c_{j,i})\R(b_{1,i})|$, $|\sum_i \R(c_{j,i})\I(b_{1,i})|$, $|\sum_i \I(c_{j,i})\I(b_{1,i})|$ we can conclude in (\ref{show2}). Similarly, to obtain, (\ref{second}), we again use bounds on real and imaginary parts. This time we consider only the nonzero entries which are at most $2np$ with high probability. Then, denoting, for $j$'th row of $\D^*$ we can write,
\beq
\nonumber\Pro( |\sum_{i\big|b_{2,i}\neq 0} \R(r_{j,i})\R(b_{2,i})|>\frac{1}{8\sqrt{2}})\leq 3\exp(-\frac{cm}{2^8p|S_2|\la^2c_0^2})
\eeq
Doing this for all components and union bounding similarly yields (\ref{second}).
\subsection{Proof of Lemma \ref{lemdualmat}}
Finally, we give the proof of Lemma \ref{lemdualmat} which is quite similar to the proof of Lemma \ref{lemdual}.
\begin{proof}
Following the notation introduced for the matrix case, we need to show if such $\SB_1,\SB_2,\SB$ exist then a certain null space condition will hold for $\Ac$ which will guarantee recovery. Let us state this condition based on the sub gradients of nuclear norm and $\ell_1$ norm: For all $\Wb\in\N(\Ac)$ if the following holds then $\A$ is the unique optimum of JBPM.
\begin{align}
\label{NSPM}
f(\Wb):=&\lambda[\R(\li\sg(\A),\Sc(\Wb)\ri)+\|\Scb(\Wb)\|_1]\\
&\hspace{10pt}+\R(\li\U\V^*,\Wb\ri)+\|\Lcb(\Wb)\|_\star>0
\end{align}
Now, assume such $\SB_1,\SB_2,\SB$ exist and consider $\vv_1,\vv_2$ where:
\begin{equation}
\vv_1=\Ac^\star(\SB_1)+\SB~~~\text{and}~~~\vv_2=\Ac^\star(\SB_2)-\SB
\end{equation}
Observe that for any $\Wb\in\N(\Ac)$, we have $\li\vv_1+\vv_2,\Wb\ri=0$. Now, using this:
\begin{align}
\label{eqzeroM}
0&=\R(\li\vv_1+\vv_2,\Wb\ri)\\
&= \R(\lambda\cdot\li\sg(\A),\Sc(\Wb)\ri+\li\Scb(\Ac^\star(\SB_2)-\SB),\Wb\ri)\nonumber\\
\nonumber&\hspace{30pt}+\R(\li\U\V^*,\Wb\ri+\li\Lcb(\Ac^*(\SB_1)+\SB),\Wb\ri)
\end{align}
To end the proof, using invertibility of $\Ac(\cdot)$ on $\Lc\cap\Sc$ we can conclude $\Lcb(\Wb)\neq 0~\text{or}~\Scb(\Wb)\neq 0$ hence:
\begin{align}
&\R(\li\Lcb(\Ac^\star(\SB_1)+\SB),\Wb\ri)<\|\Lcb(\Wb)\|_\star~~~\text{or}\\
&\R(\li\Scb(\Ac^\star(\SB_2)-\SB),\Wb\ri)<\lambda\|\Scb(\Wb)\|_\infty
\end{align}
Overall, existence of $\SB_1,\SB_2,\SB$ implies the desired null space condition, i.e., $f(\Wb)>0$ for all $\Wb\in\N(\Ac)$.
\end{proof}
%
%

\end{document}